\documentclass[letterpaper,10pt,conference]{ieeeconf}

\usepackage{cite}
\usepackage{graphicx}
\usepackage{tikz}
\usepackage{subcaption}
\usepackage[cmex10]{amsmath}
\usepackage{algpseudocode}
\usepackage{algorithmicx} 
\usepackage{array}
\usepackage[hyphens]{url}
\usepackage{mathrsfs}

\usepackage{epsfig}
\usepackage{amsfonts,amssymb,multirow,bigstrut,booktabs,ctable,latexsym}

\usepackage[mathscr]{eucal}
\usepackage{verbatim}

\usepackage{amsthm}
\usepackage{algorithm}
\usepackage{stmaryrd}

\IEEEoverridecommandlockouts  
\begin{document}

\newtheorem{definition}{Definition}
\newtheorem{lemma}{Lemma}
\newtheorem{corollary}{Corollary}
\newtheorem{theorem}{Theorem}
\newtheorem{example}{Example}
\newtheorem{proposition}{Proposition}
\newtheorem{remark}{Remark}
\newtheorem{assumption}{Assumption}
\newtheorem{corrolary}{Corrolary}
\newtheorem{property}{Property}
\newtheorem{ex}{EX}
\newtheorem{problem}{Problem}
\newcommand{\argmin}{\arg\!\min}
\newcommand{\argmax}{\arg\!\max}
\newcommand{\st}{\text{s.t.}}
\newcommand \dd[1]  { \,\textrm d{#1}  }

\title{\Large\bf Barrier Certificate based Safe Control for LiDAR-based Systems under Sensor Faults and Attacks}

\author{Hongchao Zhang, Shiyu Cheng, Luyao Niu, and Andrew Clark %
\thanks{Hongchao Zhang, Shiyu Cheng, and Andrew Clark are with the Electrical and Systems Engineering Department, McKelvey School of Engineering, Washington University in St. Louis, St. Louis, MO 63130.
{\tt\small \{hongchao,cheng.shiyu,andrewclark\}@wustl.edu}}
\thanks{Luyao Niu is with Department of Electrical and Computer Engineering, University of Washington, Seattle, WA 98195.
{\tt\small \{luyaoniu\}@uw.edu}}
\thanks{This work was supported by the National Science Foundation via grant CNS-1941670.}
}
\thispagestyle{empty}
\pagestyle{empty}

\maketitle

\begin{abstract}
Autonomous Cyber-Physical Systems (CPS) fuse proprioceptive sensors such as GPS and exteroceptive sensors including Light Detection and Ranging (LiDAR) and cameras for state estimation and environmental observation. It has been shown that both types of sensors can be compromised by malicious attacks, leading to unacceptable safety violations.
We study the problem of safety-critical control of a LiDAR-based system under sensor faults and attacks. We propose a framework consisting of fault tolerant estimation and fault tolerant control.
The former reconstructs a LiDAR scan with state estimations, and excludes the possible faulty estimations that are not aligned with LiDAR measurements. We also verify the correctness of LiDAR scans by comparing them with the reconstructed ones and removing the possibly compromised sector in the scan.
Fault tolerant control computes a control signal with the remaining estimations at each time step. 
We prove that the synthesized control input guarantees system safety using control barrier certificates.
We validate our proposed framework using a UAV delivery system in an urban environment. We show that our proposed approach guarantees safety for the UAV whereas a baseline fails.
\end{abstract}

\section{Introduction}\label{sec:intro}

Autonomous Cyber-Physical Systems (CPS) are expected to satisfy safety property in different applications \cite{knight2002safety}.
Safety violations can lead to severe economic loss and catastrophic damage to systems as well as human operators \cite{knight2002safety}. When the system can perfectly observe its state and the surrounding environment, safe control methodologies have been proposed including control barrier function (CBF) \cite{ames2016control}, Hamilton-Jacobi-Bellman-Isaacs (HJI) equation \cite{tomlin1998conflict}, and finite-state abstraction-based \cite{girard2012controller} approaches.

In real-world applications, system states and the environment are measured by sensors. As the environment becomes increasingly complex, modern CPS utilize exteroceptive sensors including Light Detection and Ranging (LiDAR) and cameras to obtain richer perception of the operating space \cite{yeong2021sensor}. Fusion among the exteroceptive sensors and proprioceptive sensors such as GPS and odometer allows CPS to better understand the environment \cite{debeunne2020review} and ensure safe operation.

Sensors have been shown to be vulnerable to faults and malicious attacks, under which ensuring CPS safety becomes more challenging. The navigation sensors can be spoofed by an adversary to cause crashes of autonomous vehicles \cite{petit2014potential,kerns2014unmanned}. Reflections \cite{tatoglu2012point} and malicious attacks \cite{cao2019adversarial,shin2017illusion} targeting LiDAR sensors can create a compromised description of the environment. These false sensor measurements bias the CPS state estimation and observations over the environment, leading CPS to make erroneous control decisions and incur safety violations.

Modeling and detection of sensor faults and attacks have been extensively studied \cite{mo2010false,liu2011false,punvcochavr2014constrained}. Secure system state estimation using measurements from proprioceptive sensors has been investigated in \cite{shoukry2017secure,fawzi2014secure}. Closed-loop safety-critical control under sensor faults and attacks has been recently studied in \cite{niu2019lqg,clark2020control}. However, these approaches are applicable to CPS using only proprioceptive sensors. When exteroceptive sensors such as LiDAR are adopted by CPS, the impact of attacks on the output of the nonlinear filters used to process LiDAR measurements are not incorporated into the aforementioned safety-critical control designs \cite{niu2019lqg,clark2020control}, rendering them  less effective. 

In this paper, we study the problem of safety-critical control for a LiDAR-based system in the presence of sensor faults and attacks. 
We propose a fault tolerant safe control framework consisting of two components, namely fault tolerant estimation and fault tolerant control. 
Our proposed framework leverages the fact \cite{cao2019adversarial} that only a narrow sector (normally within $8^\circ$) of LiDAR scans can be compromised by an adversary. Using this insight, we select system state estimations and sectors in LiDAR scans simultaneously so that they are aligned, while removing untrusted state estimates. We then use the selected state estimations and LiDAR measurements to compute a control input with safety guarantees. We make the following specific contributions:
\begin{itemize}
    \item We propose a fault tolerant state estimation algorithm that is resilient to attacks against proprioceptive sensors and LiDAR measurements. Our approach reconstructs a simulated scan based on a  state  estimate and a precomputed map of the environment. We leverage this reconstruction to remove false sensor inputs as well as detect and remove spoofed LiDAR measurements.
    
    \item We propose a fault tolerant safe control design using control barrier certificates. We present a sum-of-squares program to compute a control barrier certificate, which verifies a given safety constraint in the presence of estimation errors due to noise and attacks. We prove bounds on the probability that our synthesized control input guarantees safety.
    \item We validate our proposed framework using a UAV delivery system equipped with multiple sensors including a LiDAR. We show that the UAV successfully avoids the obstacles when navigating in an urban environment using our synthesized control law, while crashes into the unsafe region using a baseline.
\end{itemize}

The remainder of this paper is organized as follows. Section \ref{sec:related} presents the related work. Section \ref{sec:formulation} presents the system model, threat model, and necessary background. Section \ref{sec:Frame_TSV} presents our proposed fault tolerant safe control framework along with its safety guarantee. Section \ref{sec:simulation} gives a numerical case study on a UAV delivery system. Section \ref{sec:conclusion} concludes the paper.




\section{Related Work}\label{sec:related}
Ensuring CPS safety has attracted extensive research attention. Typical approaches include finite-state abstraction \cite{girard2012controller}, HJI equation \cite{tomlin1998conflict}, and counterexample-guided synthesis \cite{frehse2008counterexample}. 
Barrier function-based approaches, which formulate the safety constraint as a linear inequality over the control input, have been proposed to guarantee safety for CPS  \cite{ames2016control,usevitch2020strong,xiao2022control}. 
These approaches are applicable to CPS estimating system state using proprioceptive sensors.

Safety-critical control for systems using exteroceptive sensors such as cameras and LiDAR have been recently investigated in \cite{dean2020robust,dean2020guaranteeing,dawson2022learning,xiao2022differentiable}.
CBFs designed for high-dimensional exteroceptive sensor measurements including measurement-robust CBF \cite{dean2020guaranteeing}, observation-based neural CBF \cite{dawson2022learning}, and differentiable CBFs for learning systems \cite{xiao2022differentiable} have been proposed to compute controllers with safety guarantees.



False data injection (FDI) attacks have been reported in different applications, including modern power systems \cite{liu2011false} and unmanned aerial vehicle (UAV) \cite{wei2014simulation}. To this end, modeling, mitigating, and detecting FDI have been studied in \cite{mo2010false,liu2011false,fawzi2014secure,shoukry2017secure,hosseinzadeh2022reference}. 
LiDAR sensors have been demonstrated to be vulnerable to spoofing attacks in \cite{shin2017illusion,liu2021seeing}. The authors of \cite{cao2019adversarial} designed attacks that are capable of injecting false points at different locations in the point cloud. In \cite{khazraeiresiliency}, a stealthy attack against a perception-based controller equipped with an anomaly detector were proposed.


The existing literature on safe control in the presence of FDI attacks mainly focuses on systems with proprioceptive sensors. In \cite{niu2019lqg}, a barrier certificate based approach is proposed to ensure safety and reachability under FDI attack. A fault tolerant CBF is introduced in \cite{clark2020control} to ensure joint safety and reachability under attacks targeting proprioceptive sensors. In \cite{hallyburton2021security}, the authors have demonstrated that camera and LiDAR fusion is secure against naive attacks. For systems under attacks targeting both proprioceptive and exteroceptive sensors, how to synthesize a safety-critical control has been less studied. 


\section{Problem Formulation}\label{sec:formulation}
In this section, we introduce the system and threat model. We then formulate the problem and give needed background. 

\subsection{System Dynamics and Observation Model}
Consider a discrete-time control-affine system given as:
\begin{equation}
    \label{eq:dynamic}
    x[k+1] = f(x[k]) + g(x[k]) u[k] + w[k]
\end{equation}
where $w[k]$ is a Gaussian process with mean  zero and autocorrelation function $R_{w}(k,k^{\prime}) = Q_{k}\delta(k-k^{\prime})$ with $\delta$ denoting the discrete-time delta function and $Q_{k}$ is a positive definite matrix.  We assume that there is a nominal controller $u=\pi(x)$, for some function $\pi: \mathcal{X} \rightarrow \mathbb{R}^{m}$.
We let $x[k]\in \mathcal{X} \subseteq \mathbb{R}^n$ denote the system state and $u[k] \in \mathbb{R}^m$ denote a control signal at time $k$. Functions $f:\mathbb{R}^n\rightarrow\mathbb{R}^n$ and $g:\mathbb{R}^{n}\rightarrow\mathbb{R}^{n\times m}$ are assumed to be Lipschitz continuous.

System \eqref{eq:dynamic} uses a set of sensors $I_p:=\{1,\ldots,n_{p}\}$ to measure its states with observation $y[k] \in \mathbb{R}^{z}$ following the dynamics described as:
\begin{equation}
    \label{eq:observation}
    y[k] = o(x[k]) + v[k],
\end{equation}
where $o:\mathbb{R}^n\rightarrow\mathbb{R}^z$ is the observation function, $v[k]$ is an independent Gaussian process with mean identically zero and autocorrelation function $R_{v}[k,k^{\prime}] = R_{k}\delta(k-k^{\prime})$ and $R_{k}$ is a positive definite matrix. 

The system is equipped with a LiDAR sensor that observes the environment by calculating the ranges and angles to objects. 
A LiDAR sensor fires and collects $n_s$ laser beams to construct a scan $S:=\{(s^r_i, s^a_i),\ 0\leq i\leq n_s\}$, where $s_i^r$ denotes the range of the $i$-th scan, and $s_i^a$ denotes the angle of the $i$-th scan. 
We denote the Cartesian translated LiDAR scan $S$ measured at pose $x$ as $\mathcal{O}(x,S)$. 


We assume a 2D point-cloud map $\mathcal{M}$ is known by the system as prior knowledge. The map $\mathcal{M}:=\{(m^x_i,m^y_i),\ 0\leq i\leq n_{\mathcal{M}}\}$ is a collection of $n_{\mathcal{M}}$ points with tuples of object positions $(m^x_i,m^y_i)$ in the world coordinate.




\subsection{Threat Model}
\label{subsec:threat}
We assume that there exists an adversary that aims to cause  collisions or other unsafe behaviors.
The adversary has the capability to utilize any state-of-the-art spoofer for different sensors to conduct false data injection to perturb the observations. The injected false data denoted as $a$ can bias the system state estimation and cause the system to make incorrect control decisions.
We denote the perturbed observations as
\begin{equation}
    \label{eq:fdi_obs}
    \Bar{y}[k] = o(x[k]) + v[k] + a[k]. 
\end{equation}






The adversary can also compromise the LiDAR sensor by creating a near obstacle as demonstrated in \cite{cao2019adversarial}.
The adversary fires laser beams to inject several artificial points $e^{\prime}$ into a LiDAR scan. We denote the compromised LiDAR scan as $S \oplus e^{\prime}$, where $\oplus$ is a merge function introduced by \cite{cao2019adversarial}. 
However, due to the physical limitation of spoofer hardware, the injected point can only be within a very narrow spoofing angle, i.e. $8^{\circ}$ horizontal angle.

We index the LiDAR sensor as the $0$-th sensor and define $I=\{0\} \bigcup I_p$. 
We denote the set of sensors attacked by the adversary as $\mathcal{A}\subseteq I$. We assume that the system is uniformly observable from the  sensors in $I\backslash \mathcal{A}$. We assume that, at each time $k$, the support of $a[k]$ is contained in $\mathcal{A}$.



\subsection{Safety and Problem Formulation}



We define the state space $\mathcal{X}$ and a safety set $\mathcal{C}$ as  
\begin{equation*}
\mathcal{X} =\{x:h(x)\geq0\},\quad
\mathcal{C}= \{x\in\mathcal{X}:h_0(x)\geq0\},
\end{equation*}
where $h, h_0:\mathcal{X}\mapsto\mathbb{R}$. We say system \eqref{eq:dynamic} is safe with respect to $\mathcal{C}$ if $x[k]\in\mathcal{C}$ for all time $k= 0,1,\ldots$. 
We assume that the safe region $\mathcal{C}$ is pre-defined and known by the system, and the initial state of the system is safe, i.e. $x_0 \in \mathcal{C}$. 




\begin{problem}
\label{prob:main}
Given a map $\mathcal{M}$ and a safety set $\mathcal{C}$, we consider a nonlinear LiDAR-based system with dynamics \eqref{eq:dynamic} that is controlled by a nominal controller. The problem studied is to find a scheme to ensure system safety with desired probability $(1-\epsilon)$, where $\epsilon\in(0,1)$, when an adversary is present. 
\end{problem} 


\subsection{Preliminaries}
In what follows, we give background on discrete-time Extended Kalman Filter (EKF) and estimating pose from LiDAR scans

\subsubsection{DT-EKF}
For the system with dynamics \eqref{eq:dynamic} and observation \eqref{eq:observation}, the state estimate $\hat{x}$ is computed via EKF as:
\begin{equation}
\label{eq:dt-ekf}
    \hat{x}[k+1] = F(\hat{x}[k], u[k]) + K_k(y[k] - o(\hat{x}[k])),
\end{equation}
where $F(x[k],u[k])= f(x[k]) + g(x[k]) u[k]$.
The Kalman filter gain is 
\begin{equation}
\label{eq:ekfgain}
    K_k = A_{k} P_{k} C_{k}^{T}(C_{k} P_{k} C_{k}^{T}+R_{k})^{-1},
\end{equation}
where $A_{k} =\frac{\partial F}{\partial x}(\hat{x}[k], u[k]) $,   $C_{k} =\frac{\partial o}{\partial x}(\hat{x}[k])$, and $P_k$ is defined by the Riccati difference equation:
\begin{equation*}
    P_{k+1}=A_{k} P_{k} A_{k}^{T}+Q_{k}-K_{k}(C_{k} P_{k} C_{k}^{T}+R_{k}) K_{k}^{T}. 
\end{equation*}

The error bound of discrete-time EKF can be derived by Theorem 3.2 in \cite{reif1999stochastic} if Assumption \ref{assump:dtekf} holds.

\begin{assumption}
\label{assump:dtekf}
The system described by \eqref{eq:dynamic} and \eqref{eq:observation} satisfies the conditions: 
\begin{itemize}
    \item $A_{k}$ is nonsingular for every $k \geq 0$.
    \item There are positive real numbers $\bar{a}, \bar{c}, \underline{p}, \bar{p}>0$ such that the following bounds on various matrices are fulfilled for every $k \geq 0$ :
    \begin{equation*}
        \begin{aligned}
        &\left\|A_{k}\right\| \leq \bar{a}; \
        \left\|C_{k}\right\| \leq \bar{c}; \
        \underline{p} I \leq P_{k} \leq \bar{p} I; \\
        &\underline{q} I \leq Q_{k}; \ 
        \underline{r} I \leq R_{k} .
        \end{aligned}
    \end{equation*}
    \item Let $\phi$ and $\chi$ be defined as
    \begin{align*}
        &F(x[k],u[k]) - F(\hat{x}[k],u[k]) = A_k (x[k]-\hat{x}[k]) \\
        &\quad\quad\quad\quad\quad\quad\quad\quad+ \varphi(x[k], \hat{x}[k], u[k]) \\
        &o(x[k])-o(\hat{x}[k]) =C_{k}(x[k]-\hat{x}[k])+\chi(x[k], \hat{x}[k])
    \end{align*}
    Then there are positive real numbers $\epsilon_{\varphi}, \epsilon_{\chi}, \kappa_{\varphi}, \kappa_{\chi}>0$ such that the nonlinear functions $\varphi, \chi$ are bounded via
    \begin{equation*}
        \|\varphi(x, \hat{x}, u)\|  \leq \kappa_{\varphi}\|x-\hat{x}\|^{2},\quad
    \|\chi(x, \hat{x})\|  \leq \kappa_{\chi}\|x-\hat{x}\|^{2}
    \end{equation*}
    for $x, \hat{x} \in R^{n}$ with $\|x-\hat{x}\| \leq \epsilon_{\varphi}$ and $\|x-\hat{x}\| \leq \epsilon_{\chi}$, respectively.
\end{itemize}

\end{assumption}

If the conditions of Assumption 1 hold, the estimation error $\zeta_{k}=x[k]-\hat{x}[k]$ is exponentially bounded in mean square and bounded with probability one, provided that the initial estimation error satisfies $\left\|\zeta_{0}\right\| \leq \Bar{\zeta}$ \cite{reif1999stochastic}. 

\subsubsection{Estimating Pose By Comparing Scans}
Pose refers to the position of the system in a Cartesian coordinate frame. Pose estimations with LiDAR scans have been extensively studied. NDT~\cite{biber2003normal}, as one of the widely-used approaches, models the distribution of all reconstructed 2D-Points of one laser scan by a collection of local normal distributions. 

Consider two states $x_1,x_2\in\mathcal{X}$ and the LiDAR scans $\mathcal{O}(x_1,S_1)$ and $\mathcal{O}(x_2,S_2)$ collected at $x_1$ and $x_2$, respectively. The NDT method estimates the relative pose change as $r=\mathcal{O}(x_1,S_1) \ominus \mathcal{O}(x_2,S_2)$, where $\ominus$ is a scan match operation.
The scan match operation is implemented as follows. The NDT method first subdivides the surrounding space uniformly into cells with constant size. For each cell in $\mathcal{O}(x_1,S_1)$, the mean $q$ and the covariance matrix $\Sigma$ are computed to model the points contained in the cell as the normal distribution $N(q, \Sigma)$. 
Denote the points in $\mathcal{O}(x_2,S_2)$ as $p_i,\ i\in n_s$, where $p_i$ is a position vector and $n_s$ is the number of valid points.
Define loss function $\mathcal{L}_s(r^{\prime})$ as
\begin{equation}
    \label{eq:score}
    \mathcal{L}_s(r^{\prime}) = \sum_{i} \exp \left(\frac{-((p_i-r^{\prime})-{q}_{i})^{T} {\Sigma}_{i}^{-{1}}((p_i-r^{\prime})-{q}_{i})}{2}\right)
\end{equation}
The relative pose change $r^{\prime}$ is estimated by solving the minimization problem
\begin{equation}
\label{eq:ndt_obj}
    \min_{r^{\prime}}{-\mathcal{L}_s(r^{\prime})}
\end{equation}
with Newton's algorithm. We use $r$ to denote the solution to \eqref{eq:ndt_obj} for the rest of the paper. The corresponding loss $\mathcal{L}_s(r)$ can be computed with the output of scan match $r$ by \eqref{eq:score}. 

\section{Fault Tolerant   Safe  Control Framework}\label{sec:Frame_TSV}

In this section, we propose a framework for safe control that is compatible with  existing LiDAR-based autonomous systems. We first give an overview and then describe each component in detail.


\subsection{Overview of Framework}
\label{subsec:FT-Frame}

We consider a system with dynamics \eqref{eq:dynamic} and observation model \eqref{eq:observation} in the presence of an adversary, as described in Section \ref{sec:formulation}.
To guarantee the system's safety under attacks, we propose a fault tolerant framework to ensure safety at each time step. The framework consists of two parts, namely \emph{fault tolerant estimation} and \emph{fault tolerant control}.

The idea of fault tolerant estimation is to exclude compromised sensors in $I_p$ by utilizing additional information contained in LiDAR sensor measurements. 
We maintain a set of state estimations $\hat{x}_i$ using EKF, where $i \in I_l\subseteq 2^{I_p}$ and each element of $i \in I_l$ is a collection of sensors in $I_p$ such that system \eqref{eq:dynamic} is uniformly observable from the sensors in $I_{l}$. 
As shown in Fig. \ref{fig:FT-est}, a fault tolerant estimation reconstructs a LiDAR observation, denoted as $\mathcal{O}(\hat{x}_i,\mathcal{M})$, for each state estimation $\hat{x}_i$. The reconstruction is achieved by simulating the scan process on knowledge map $\mathcal{M}$ with state estimate $\hat{x}$ being the center. 
We propose a fault tolerant LiDAR estimation to compare the estimated LiDAR scan $\mathcal{O}(\hat{x}_i,\mathcal{M})$ with the actual LiDAR measurement $\mathcal{O}(x,S)$. The comparison then provides a pose estimation. Using the pose estimation, our proposed fault tolerant state estimation excludes the conflicting state estimations, i.e., the state estimations that deviate from the LiDAR estimation. 

After excluding the conflicting state estimations using fault tolerant estimation, we then design fault tolerant safe control to ensure safety of the system at each time step. Fault tolerant safe control computes an input $u_o$ that does not deviate too far from the nominal controller $\pi(\hat{x}_i)$ for all $i$ given by the fault tolerant estimation. The safety of $u_o$ is certified by a discrete-time barrier certificate. 

In what follows, we describe the fault tolerant estimation in two-fold, that is fault tolerant LiDAR estimation (Section \ref{sec:ft-LiDAR est}) and fault tolerant state estimation (Section \ref{sec:ft-state est}).

\begin{figure}[]
    \centering
    \includegraphics[width=0.45\textwidth]{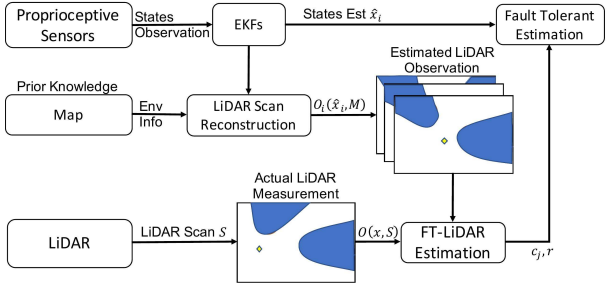}
    \caption{Fault tolerant estimation for LiDAR-based system removes conflicting state estimations by comparing estimations of proprioceptive sensors with additional information from exteroceptive sensors measurements. }
    \label{fig:FT-est}
\end{figure}




\begin{algorithm}[b]
\caption{LiDAR Scan Reconstruction}
\begin{algorithmic}[1]
    \State \textbf{Input:} State estimate $\hat{x}_i$, point-cloud map $\mathcal{M}$  
    \State \textbf{Parameters:}  Resolution of the LiDAR scan $c_{r}$, maximum LiDAR range $r_{max}$.  
    \State \textbf{Output:} Estimated LiDAR Observation $\mathcal{O}(\hat{x}_i,\mathcal{M})$
    \State \textbf{Init:} Set $\hat{x}_i$ as the center of scan $S_{\mathcal{M}}$, set $l_k^r\gets r_{max}$. Separate the scan equally into $\frac{2\pi}{c_{r}}$ sectors $S_{k}$ with corresponding angle $l_k^a$. 
    \State Translate points $m_j\in \mathcal{M}$ into polar coordinate with the origin $\hat{x}_i$, and represent it with a tuple $(m_j^r, m_i^a)$.
    \For{$m_j \in \mathcal{M}$ and $k\in [0,c_r]$}
            \If{$m_j \in S_{k}$ and $m_j^r\leq l_k^r$}
                \State $l_k^r \gets p_i^r$
            \EndIf
    \EndFor
    \For{$k \ \st \ l_k^r= r_{max}$}
        \State $l_k^r\gets NaN$
    \EndFor
    \State Reconstruct $S_{\mathcal{M}}=\{(l_k^r, l_k^a)\}$
    \State \textbf{Return} $\mathcal{O}(\hat{x}_i,\mathcal{M})=\mathcal{O}(\hat{x}_i,S_{\mathcal{M}})$
\end{algorithmic}
\label{alg:ScanReconstruct}
\end{algorithm}

\subsection{Fault Tolerant LiDAR Estimations}\label{sec:ft-LiDAR est}


In the following, we introduce fault tolerant LiDAR estimation. This procedure converts each state estimation $\hat{x}_i$ given by EKFs to an estimated LiDAR observation $\mathcal{O}(\hat{x}_i,\mathcal{M})$ using map $\mathcal{M}$. The estimated LiDAR observation is then compared with the actual LiDAR observation to exclude possible faults in state estimations.

Fault tolerant LiDAR estimation is presented in Alg. \ref{alg:ScanReconstruct}. Given parameters on the resolution of the LiDAR scan $c_{r}$ and maximum LiDAR range $r_{max}$, we initialize the estimated scan $S_{\mathcal{M}}=\{(l_k^r, l_k^a)\}$ with a circle centered at $\hat{x}_i$ and radius as $r_{max}$. 
We equally divide the circle and assign sectors $S_k$ to the corresponding $l_k^a$. 
Next, we represent the points in map $\mathcal{M}$ using polar coordinates with the origin at $\hat{x}_i$. To simulate the scan, we assign the closest point to the scan from line 6 to 10. We iterate through all points $m_j\in \mathcal{M}$. For the point  in sector $S_k$, we replace $l_k^r$ with $m_j^r$ if $m_j^r\leq l_k^r$. 
Then we remove the points that have never been updated. Finally, we output estimated observation $\mathcal{O}(\hat{x}_i,\mathcal{M})=\mathcal{O}(\hat{x}_i,S_{\mathcal{M}})$. Intuitively, this estimated observation can be viewed as the output of a LiDAR scan centered at state $\hat{x}_{i}$ with object locations given in the map $\mathcal{M}$. Hence, any deviation of the estimated and actual scans indicates either an error in the state estimate or a spoofing attack on the scan. 





Next, we consider the case where the adversary not only injects false data into prorioceptive sensors but also spoofs LiDAR sensors. 
The intuitive countermeasure is to remove the region of the scan  that is impacted by false data. Since the adversary is only capable of modifying points in the scan within a narrow spoofing angle, our approach is to partition the scan and map into regions $c_{j}$ and attempt to identify which region has been impacted by spoofing. That region is then removed from the scan and the estimated scan. 
Since the adversary tries to bias the state estimation, we model the problem of choosing a set of observations to ignore in order to mitigate the impact of false data  as a minimax optimization
\begin{equation}
    \min_{e^{\prime}_I}{\max_{c_j} {\mathcal{L}_s(\tilde{r})}}
\end{equation}
where $\tilde{r} = {\mathcal{O}(\hat{x}_{i}, \mathcal{M}\backslash c_j)\ominus \mathcal{O}(x,S\oplus e^{\prime}_I\backslash c_j)}$. 
We search for subdivision $c_j$ through the LiDAR observation space with Alg. \ref{alg:FT-LiDAR_Est}, which is detailed as follows.

\begin{algorithm}[b]
\caption{FT-LiDAR Estimation}
\begin{algorithmic}[1]
    \State \textbf{Input:} State estimation $\hat{x}$, number of sector $n_{j}$, Map $\mathcal{M}$ and LiDAR scan $S$
    \State \textbf{Output:} $r_j, c_j$
    \State \textbf{Init: } Equally separate scan $S$ into $n_{j}$ sectors $c_j\in S$
    \For{$c_j\in S$}
        \State Scan Reconstruction $\mathcal{O}_j(\hat{x},\mathcal{M}\backslash c_j)$
        \State Scan Reconstruction $\mathcal{O}_j(x,S\backslash c_j)$
        \State Compute $n_s^j$ the number of points in $S\backslash c_j$. 
        \State Compute $\tilde{r}_j=\mathcal{O}_j(\hat{x},\mathcal{M}\backslash c_j) \ominus \mathcal{O}_j(x,S\backslash c_j)$
        \State Compute $\zeta_s^j = n_s^j - \mathcal{L}_s(r_j)$
        \If{$\zeta_s^j\leq \bar{\zeta}_s$}
            \Return $\tilde{r}_j, c_j$
        \EndIf
    \EndFor
    
\end{algorithmic}
\label{alg:FT-LiDAR_Est}
\end{algorithm}

The adversary compromises the LiDAR scan $S$ by merging it with false data $e'_I$, denoted as $S\oplus e^{\prime}_I$. 
As shown in Alg. \ref{alg:FT-LiDAR_Est}, we take in state estimation $\hat{x}_i$, number of sectors $n_{j}$, map $\mathcal{M}$, and scan $S$ to search for sector $c_j$ over scan $S$. The algorithm outputs the corresponding estimated relative pose $\tilde{r}_j$. 
For each sector $c_j$, we estimate observations $\mathcal{O}(\hat{x}_{i}, \mathcal{M}\backslash c_j)$ with Alg. \ref{alg:ScanReconstruct} and reconstruct the corresponding LiDAR observation $\mathcal{O}(x,S\oplus e^{\prime}_I\backslash c_j)$. 
Next, we compute $n_s^j$, the number of points contained in $S\backslash c_j$, and perform scan match to obtain $\tilde{r}$ by
\begin{equation}
    \label{eq:ndt_lidspoof}
    \tilde{r}_j = \mathcal{O}(\hat{x}_{i}, \mathcal{M}\backslash c_j)\ominus \mathcal{O}(x,S\oplus e^{\prime}_I\backslash c_j). 
\end{equation}
Then, we compute the loss function $\mathcal{L}_s(\tilde{r})$ and the performance degradation $\zeta_s^j = n_s^j - \mathcal{L}_s(\tilde{r})$. 
Finally, we output $\tilde{r}_j$ and $c_j$ for $\zeta_s^j\leq \bar{\zeta}_s$. 
In what follows, we compute the upper bound $\bar{\zeta}_s$ of the degradation of the loss $\mathcal{L}_s$ brought by noise as the criteria of whether LiDAR sensor is affected by factors other than noise. 
We consider a point $p_i$ sampled in the LiDAR scan collected at state $x$ with a zero-mean disturbance $w_i$ whose norm is bounded as $\|w_i\|\leq \bar{w}_i$. 

\begin{theorem}
\label{th:scanmatch}
    Consider a state $x$ and its state estimation $\hat{x}$. 
    Let $\mathcal{O}(x,S)$ and $\mathcal{O}(\hat{x},\mathcal{M})$ be LiDAR scan and estimated LiDAR observation. 
    Let $r=\mathcal{O}(\hat{x}_{i}, \mathcal{M})\ominus \mathcal{O}(x,S)$ and $\tilde{r}$ be computed by \eqref{eq:ndt_lidspoof} when adversary present. 
    In the case where the LiDAR sensor is not attacked, we have the performance degradation $\zeta_s$ is bounded by 
    \begin{multline}
    \label{eq:scannoise}
        \zeta_s := \mathcal{L}_s^{max}(r)-\mathcal{L}_s(r) \\
        \leq n_s - \sum_{i} \exp \left(\frac{-\bar{w}^2_i \lambda({\Sigma}_{i}^{-{1}})}{2} \right) =: \bar{\zeta_s},
    \end{multline} 
    where $\mathcal{L}_s^{max}(r)$ is the maximum of \eqref{eq:score}, $n_s$ is the number of points contained in $S$, and $\lambda({\Sigma}_{i}^{-{1}})$ is the maximum eigenvalue of $\Sigma_{i}^{-{1}}$. 
    
    When the LiDAR sensor is attacked, if a subdivision $c_j\supseteq e^{\prime}_I$ can be found by Alg. \ref{alg:FT-LiDAR_Est}, we have the performance degradation of scan match is bounded as \eqref{eq:scannoise}, where $n_s$ is the number of points contained in $S\backslash c_j$ and the summation is over all points in $S \setminus c_j$. 
\end{theorem}
\begin{proof}
We first show that $L_s^{max}(r) = n_s$. Then we derive a lower bound for $\mathcal{L}_s(r)$. 
Since covariance $\Sigma_i$ is positive definite, using \eqref{eq:score} we have
\begin{multline*}
    L_s^{max}(r)\\
    =\sum_i\exp \left(\frac{-((p_i-r)-{q}_{i})^{T} {\Sigma}_{i}^{-{1}}((p_i-r)-{q}_{i})}{2}\right) \\
    \leq\sum_i \exp{(0)} = n_s.
\end{multline*}
Let $p_i$ be a point sampled in LiDAR scan. We have that $((p_i-r)-{q}_{i})\leq =w_i$ with $w_i$ being the realized disturbance when sampling $p_i$. Since $\|w_i\|\leq \bar{w}_i$ and $\Sigma_i^{-1}$ is Hermitian, we then have
\begin{multline*}
    \sum_{i} \exp \left(\frac{-((p_i-r)-{q}_{i})^{T} {\Sigma}_{i}^{-{1}}((p_i-r)-{q}_{i})}{2}\right) \\
    \geq \sum_{i} \exp \left(\frac{-\bar{w}^2_i \lambda({\Sigma}_{i}^{-{1}})}{2} \right).
\end{multline*}
Hence, we have that $\zeta_s$ is bounded as \eqref{eq:scannoise}.
 

When the LiDAR sensor is spoofed, there always exists a subdivision $c_j$ such that the false data $e_I'$ satisfies $e^{\prime}_I\subseteq c_j$. If $c_j$ is successfully identified by Alg. \ref{alg:FT-LiDAR_Est}, then the subdivision $c_j$ along with the false data $e_I'$ are ignored. In this case, our analysis for the scenario where the LiDAR sensor is not attacked can be applied, yielding the bound in \eqref{eq:scannoise} with $n_s$ being the number of points contained in $S\setminus c_j$. If $c_j$ containing $e_I'$ is not identified and is not ignored, then by line 10 of Alg. \ref{alg:FT-LiDAR_Est}, we have that $\zeta_s^j\leq \bar{\zeta}_s$ and thus the bound in \eqref{eq:scannoise} follows.
\end{proof}

\subsection{Fault Tolerant State Estimation}\label{sec:ft-state est}

We next propose the criteria to develop an algorithm for a fault tolerant state estimation that provides bounded estimation error under false data attacks on the proprioceptive sensors. Our approach computes a set of indices $I_{a} \subseteq I_{l}$ that are removed to ensure that the state estimation error is bounded. 
Given a state estimation deviation threshold $\theta_h$ and a scan match degradation threshold $\bar{\zeta}_s$, a state estimate is \emph{not removed} (i.e. $i \notin I_{a}$) if either of the following criteria holds: 
\begin{itemize}
    \item Case I: $i \notin I_a$ for estimation indexed $i\in I_l$, if $\|r_i\|\leq \theta_h$ and $\zeta^i_s \leq \bar{\zeta}_s$.
    \item Case II: $i \notin I_a$ for estimation indexed $i\in I_l$, if $\|\tilde{r}_i\|\leq \theta_h$ and $\tilde{\zeta}_s^i \leq \bar{\zeta}_s$. 
\end{itemize}
We consider LiDAR observation is trusted, if for all $i\in I_l$ estimated LiDAR observation, the scan match degradation $\zeta^i_s \leq \bar{\zeta}_s$.
In Case I, we have the scan match degradation $\zeta^i_s \leq \bar{\zeta}_s$, and the pose deviation $\|r_i\|\leq \theta_h$. We draw the conclusion that $\hat{x}_{i}$ agrees with the LiDAR observation, and hence $i \in I\backslash I_a$. 
When the LiDAR observation is not trusted, we reconstruct estimated and actual LiDAR observation with Alg. \ref{alg:FT-LiDAR_Est} to exclude section $c_j$.
In Case II, we have the reconstructed scan match degradation $\|\tilde{r}_i\|\leq \theta_h$, and the pose deviation within tolerance with $\tilde{\zeta}_s^i \leq \bar{\zeta}^s$. We draw the conclusion that $i \in I\backslash I_a$. 

In what follows, we show that sensor $i\in I\backslash I_a$ selected by criteria is attack-free and we can further have the deviation of FT-Estimation bounded by the EKF error bound of selected sensors. 
\begin{theorem}
\label{th:ftest}
    Given scan match results $r_i$, $\tilde{r}_i$ and $\bar{\zeta}_s$, for sensor $i\in I\backslash I_a$ given by criteria I and II, we have estimation error bounded as $\|x-\hat{x}_i\|\leq \bar{\zeta}_i$.
\end{theorem}
\begin{proof}
    We prove by contradiction. We suppose that there exists a sensor $b\in I\backslash I_a$, whose estimation $\hat{x}_b$ satisfies $\|x-\hat{x}_b\|> \bar{\zeta}_b$. 
    We next show contradictions for Case I and II.

    
    In Case I, set $\theta_h=\min_i{\bar{\zeta}_i}$. Since $\zeta^b_s \leq \bar{\zeta}_s$, we have that LiDAR scan matches with estimated scan with relative pose change $r=x-\hat{x}_b$. If sensor $b$ is included in $I\backslash I_a$, we have $\|x-\hat{x}_b\|\leq \theta_h\leq \bar{\zeta_b}$, which contradicts to $\|x-\hat{x}_b\|> \bar{\zeta}_b$.
    
    In Case II, set $\theta_h=\min_i{\bar{\zeta}_i}$. Since $\tilde{\zeta}^b_s \leq \bar{\zeta}_s$, we have that LiDAR scan matches with estimated scan with relative pose change $\tilde{r}=x-\hat{x}_b$. If sensor $b$ is included in $I\backslash I_a$, we have $\|x-\hat{x}_b\|\leq \theta_h\leq \bar{\zeta_b}$, which contradicts to $\|x-\hat{x}_b\|> \bar{\zeta}_b$.
    

    Otherwise, sensor $b$ will be excluded into set $I_a$ and hence for any sensor $i\in I\backslash I_a$ we have the error bounded. 
\end{proof}
\subsection{Fault-Tolerant Safe Control}\label{sec:CBC}


We next present the fault tolerant control synthesis to ensure safety of the system.  
We set the state estimation as $\hat{x}_{\alpha}[k]=\hat{x}_i$, for some $i\in I\backslash I_a$.
We define the control input signal as $u_o[k]=\pi(\hat{x}_\alpha[k])+\hat{u}[k]$. In what follows, we assume the nominal controller is of the form $\pi(x) = \pi_{0} + K_{c}\hat{x}_{\alpha}$ for some $\pi_{0} \in \mathbb{R}^{m}$ and matrix $K_{c}$. Since we have $||x[k]-\hat{x}_{\alpha}[k]|| \leq \bar{\zeta}_{\alpha}$ by Theorem \ref{th:ftest}, the nominal control input for the estimated state satisfies $$||\pi(\hat{x}_{\alpha}[k])-\pi(x[k])|| = ||K_{c}(x[k]-\hat{x}_{\alpha}[k])|| \leq ||K_{c}||\bar{\zeta}_{\alpha}.$$
Hence, if we choose $u_{o}[k]$ such that $||\hat{u}[k]||_{2} \leq \xi - ||K_{c}||\bar{\zeta}_{\alpha}$ for some $\xi \geq 0$, then we can guarantee that the chosen control input is within a bounded distance of the nominal control input corresponding to the true state value.

\begin{proposition}
\label{prop:discrete_bc}
    Consider a discrete-time system described by \eqref{eq:dynamic} and sets $\mathcal{C}, \mathcal{D} \subseteq \mathcal{X}$. If there exist a function $B: \mathcal{X} \rightarrow \mathbb{R}_{0}^{+}$, a constant $c \geq 0$, a linear controller $u = K_c x$, and a constant $\gamma \in[0,1)$ such that
    \begin{eqnarray}
        B(x) \leq \gamma, & \  \forall x \in \mathcal{C} \label{eq:prop1_1} \\
        B(x) \geq 1, & \  \forall x \in \mathcal{D} \label{eq:prop1_2}
        \\
        \begin{aligned}
            \mathbb{E}[B(f&(x) + g(x) (K_c x + \hat{u}) \\ &+ w) \mid x]  \leq B(x)+c,  \label{eq:prop1_3}
        \end{aligned}
        & \ \forall x \in \mathcal{X}, \forall \|\hat{u}\|\leq \xi
    \end{eqnarray}
    then for any initial state $x_{0} \in \mathcal{C}$, we have the $Pr(x[k]\in \mathcal{C}, 0\leq k\leq T_d)\geq 1-\gamma- c T_d$. 
\end{proposition}

\begin{proof}
    We have $u_o - u = K_c x - K_c \hat{x}_\alpha + \hat{u}$. Since $\|\hat{u}\|\leq \xi- K_c \bar{\zeta}_\alpha$ and $\|K_c x - K_c \hat{x}_{\alpha}\|\leq K_c \bar{\zeta}_\alpha$. By triangle inequality, we can have $\|u_o-u\|\leq \xi$. 
    Since there exists a function $B(x)$ satisfying \eqref{eq:prop1_1} to \eqref{eq:prop1_3}, $B(x)$ is a control barrier certificate for system \eqref{eq:dynamic}. According to \cite{jagtap2020formal} and \eqref{eq:prop1_2}, we have
    \begin{multline*}
        Pr \left\{x[k] \in \mathcal{D}\right.\text{for some }\left.0 \leq k<T_{d} \mid x(0)=x_{0}\right\} \\
        \leq Pr\left\{\sup _{0 \leq k<T_{d}} B(x[k]) \geq 1 \mid x(0)=x_{0}\right\} \\
        \leq B\left(x_{0}\right)+c T_{d} \leq\gamma+c T_{d}. 
    \end{multline*}
\end{proof}
We define  $h_{B}^{\xi}(\hat{u})=(\xi-K_c\bar{\zeta}_\alpha)^{2}-\| \hat{u}\|_2^{2}$. The system has continuous state space $\mathcal{X}$ and action space $U$, we can follow the standard procedure to compute control barrier certificate $B(x)$ by solving an SOS programming given as follows:

\begin{proposition}
\label{prop:bcsos}
Suppose there exist a function $B(x)$ and polynomials $\lambda_0(x)$, $\lambda_1(x)$, $\lambda_x(x,\hat{u})$ and $\lambda_{\hat{u}}(x,\hat{u})$ such that
\begin{align}
    -B(x)-\lambda_{0}(x) h_{0}(x)+\gamma &\text{ is SOS} \label{eq:prop2_1}\\
    B(x)+\lambda_{1}(x) h_{0}(x)-1 &\text{ is SOS} \label{eq:prop2_2}\\
    -\mathbb{E}[B(f(x) + g(x) (K_c x  + \hat{u}) + w) \mid x] +  \nonumber\\
    B(x)-\lambda(x)h(x) -\lambda_{\hat{u}}(x, \hat{u}) h_{B}^{\xi}(\hat{u}) +c &\text{ is SOS} \label{eq:prop2_3}
\end{align}
then for any initial state $x_{0} \in \mathcal{C}$, we have the $Pr(x[k]\in C, 0\leq k\leq T_d)\geq 1-\gamma- c T_d$.
\end{proposition}

\begin{proof}
Since the entries $B(x)$ and $\lambda_{0}(x)$ in $-B(x)-\lambda_{0}(x) h_{0}(x)+\gamma$ are SOS, we have $0 \leq B(x)+\lambda_{0}(x) h_{0}(x) \leq \gamma$. Since the term $\lambda_{0}(x) h_{0}(x)$ is nonnegative over $\mathcal{C}$, \eqref{eq:prop2_1} and \eqref{eq:prop2_2} implies \eqref{eq:prop1_1} and \eqref{eq:prop1_2} in Proposition \ref{prop:discrete_bc}. 
Since the terms $\lambda_{\hat{u}}(x) h_{B}^{\xi}(\hat{u})$ and $\lambda(x) h(x)$ are nonnegative over set $\mathcal{X}$, we have \eqref{eq:prop1_3} holds, which implies that the function $B(x)$ is a control barrier certificate.
\end{proof}

The choice of $\xi$ uses a similar approach as \cite{niu2019lqg} by solving the SOS program offline to enhance the scalability. Other numerical issues of SOS program such as sparsity and ill-conditioned problem are investigated in \cite{cotorruelo2021reference,kojima2005sparsity}.

We propose Alg. \ref{alg:FTControl} to compute feasible control inputs to ensure safety at each time-step $k$. We initialize $I_a\gets \emptyset$ and define $\Omega_{i\in I\backslash I_a}:= \{ u_o: (u_o-u_i)^T (u_o-u_i) \leq \xi\}$. At each time-step $k$ we maintain $n_l$ state estimations for sensors in $I_l$ and compute control input $u_i:=\pi(\hat{x}_{i})$ with a nominal controller. 
We compute $u_o$ by solving \eqref{eq:qcqp}, where $J$ is a cost function. If no such $u_o$ exists, we perform Alg. \ref{alg:FT-LiDAR_Est} and fault tolerant state estimation to remove conflicting sensors.
\begin{algorithm}[hb]
\caption{Fault Tolerant Control}
\begin{algorithmic}[1]
    \State \textbf{Init: } $I_a\gets \emptyset$ and 
    $\Omega_{i\in I\backslash I_a}:= \{ u_o: (u_o-u_i)^T (u_o-u_i) \leq \xi\}$
    \State Maintain $n_l$ EKFs for each sensor to estimate state $\hat{x}_i, \ i\in I_l = \{1,2,\ldots, n_l\}$.
    \State Compute control input $u_i:=\pi(\hat{x}_{i})$. 
    \If{control input $u \in \bigcap_{i\in I\backslash I_a} \Omega_{i}$}
        \State set $\hat{u}=0$ and $u_o = u+\hat{u}$
    \Else
        \Comment{STEP 1} 
        \State Compute control input $\hat{u}$ such that $u_o:=u+\hat{u}$ is the solution to the following problem. 
        \begin{equation}
        \label{eq:qcqp}
            {\min_{u_o} J(\hat{x}_{i},u_o)} \
            s.t. \  u_o \in \cap_{i\in I\backslash I_a} \Omega_{i}
        \end{equation}
        \If{no such $u_o$ can be found}
        \Comment{STEP 2}
            \State Perform FT-LiDAR Estimation (Alg. \ref{alg:FT-LiDAR_Est}).
            \State Exclude false sensors into $I_a$ by criteria I and II.
            \State Compute $\hat{u}$ by solving \eqref{eq:qcqp}.
            \If{no such $u_o$ can be found}
                \Comment{STEP 3} 
                \For{$u\notin \bigcap_{i \in I\backslash I_a} \Omega_{i}$}
                    \State Compute residue values $y_{i}-o (\hat{x}_{i})$
                    \State Include $i$ into $I_a$ with the largest residue. 
                \EndFor
            \EndIf
        \EndIf
    \EndIf
    
\end{algorithmic}
\label{alg:FTControl}
\end{algorithm}

\begin{figure*}[!htbp]
\centering
\begin{subfigure}{.24\textwidth}
  \centering
  \includegraphics[width = \textwidth]{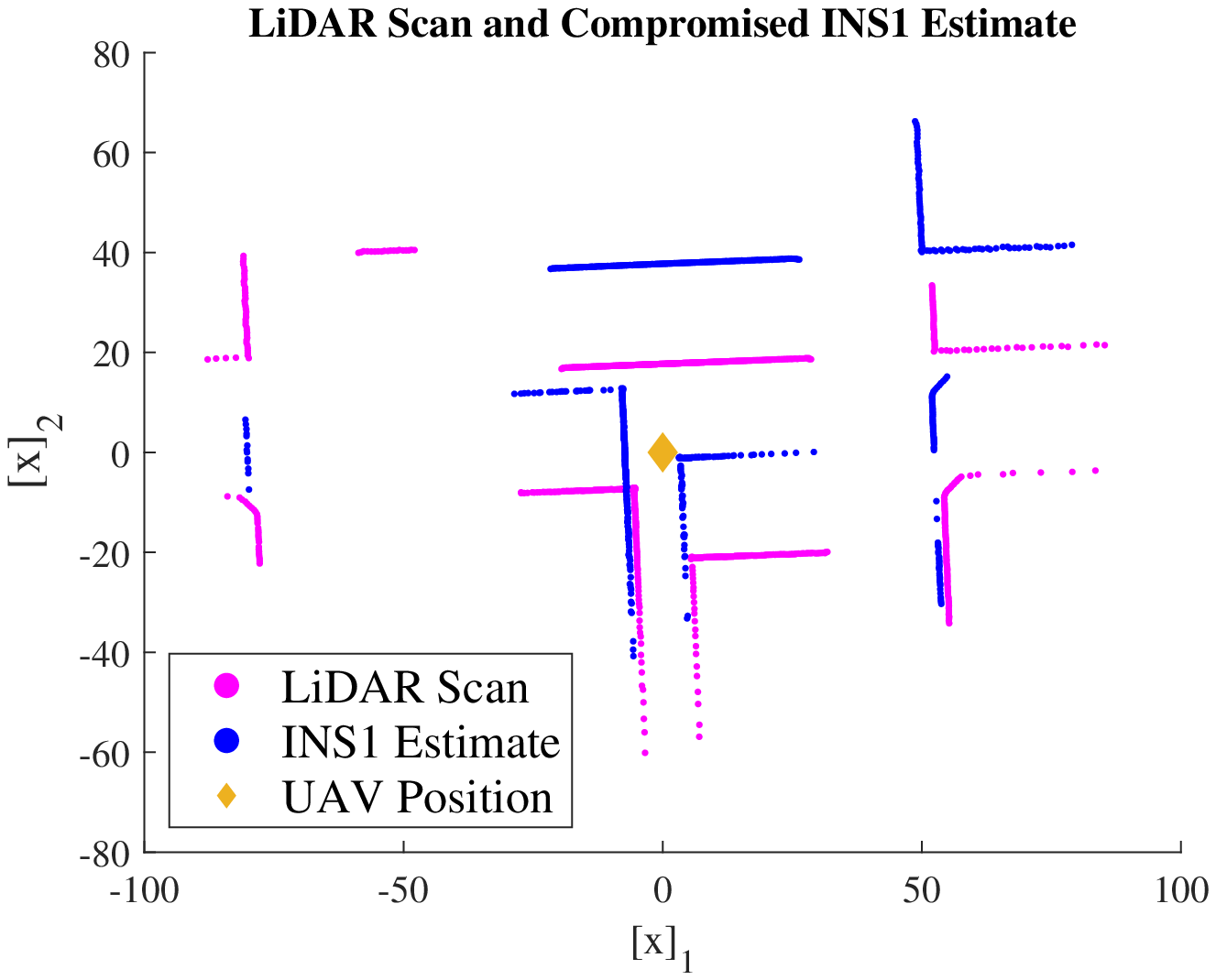}
  \subcaption{}
  \label{fig:s1x1}
\end{subfigure}%
\hfill
\begin{subfigure}{.24\textwidth}
  \centering
  \includegraphics[width = \textwidth]{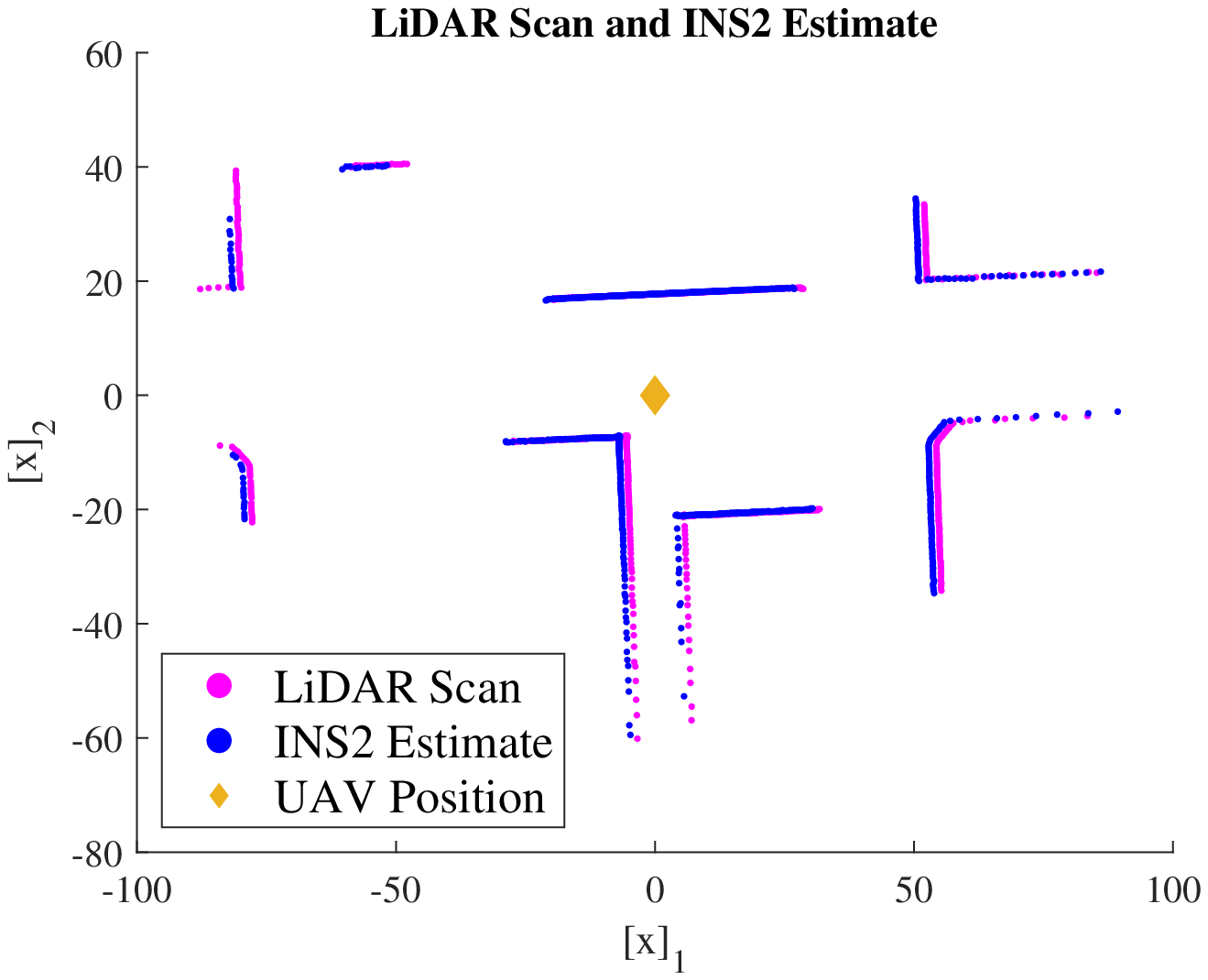}
  \subcaption{}
  \label{fig:s1x2}
\end{subfigure}%
\hfill
\begin{subfigure}{.24\textwidth}
  \centering
  \includegraphics[width=\textwidth]{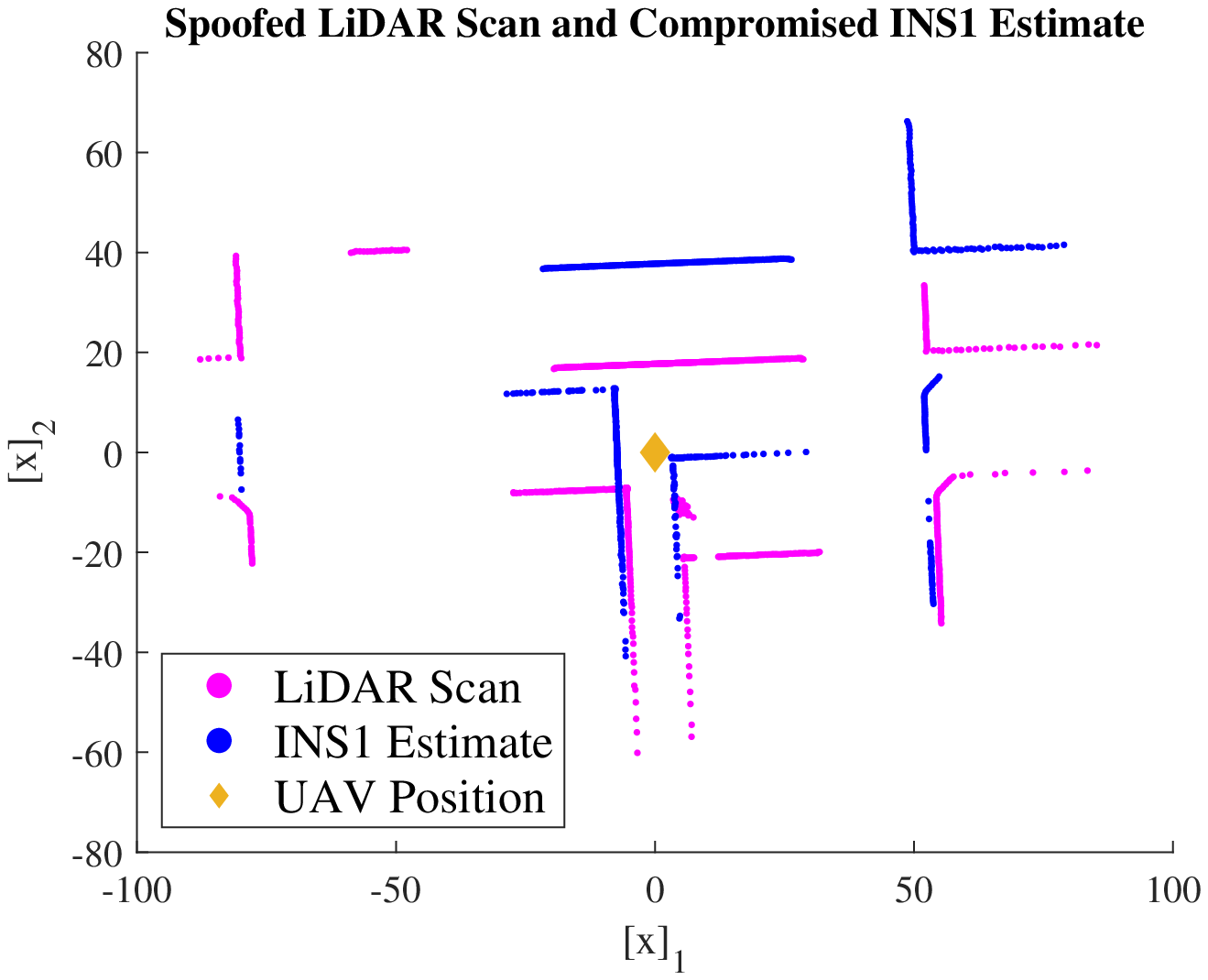}
  \subcaption{}
  \label{fig:s2x1}
\end{subfigure}%
\hfill
\begin{subfigure}{.24\textwidth}
  \centering
  \includegraphics[width=\textwidth]{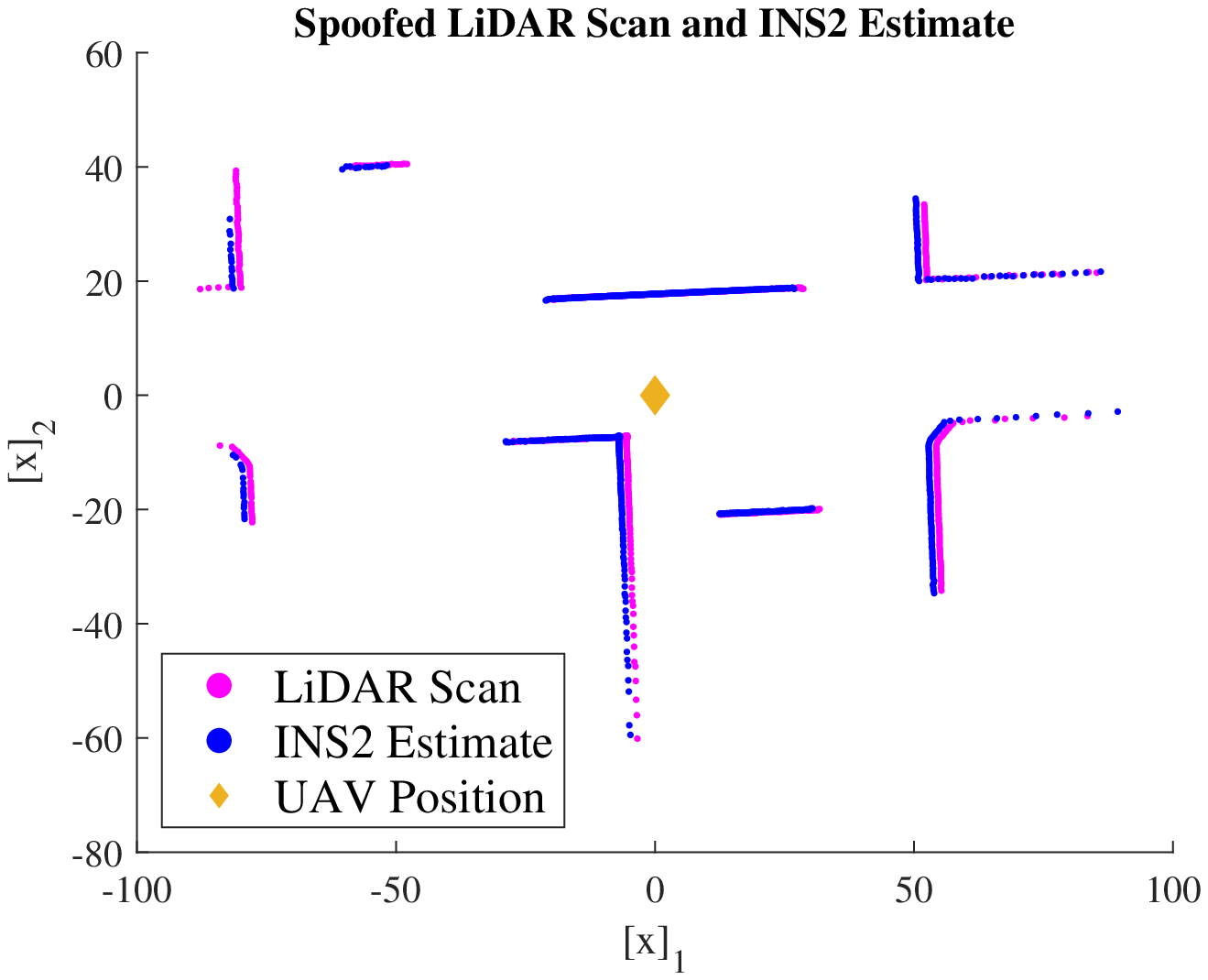}
  \subcaption{}
  \label{fig:s2x2}
\end{subfigure}%
\caption{Comparison between the estimated LiDAR observations (blue lines) and actual LiDAR observations (pink lines). Fig. \ref{fig:s1x1} to \ref{fig:s1x2} compares the estimated and actual LiDAR observations under attack Scenario I (INS1 compromised). The estimate based on INS1 deviates from the actual scan, causing the compromised sensor INS1 to become untrusted.  Fig. \ref{fig:s2x1} to \ref{fig:s2x2} compares the estimated and actual LiDAR observations under attack Scenario II (INS1 and LiDAR compromised).
Fig. \ref{fig:s1x1} and Fig. \ref{fig:s2x1} estimate the LiDAR scan using the compromised measurements from INS1. Fig. \ref{fig:s1x2} and Fig. \ref{fig:s2x2} estimate the LiDAR scan using the measurements from INS2. The proposed approach removes the spoofed obstacle and aligns with the non-compromised sensor INS2.}
\end{figure*}

\begin{theorem}
Given a safe set $\mathcal{C}$ and $\bar{\zeta}_s$, if the following conditions hold: (i) Assumption 1 holds, and (ii) scan match results $r$ and $\tilde{r}$ can be found at each time step $k$, and (iii) there exists a function $B(x)$ satisfying the conditions in Proposition \ref{prop:discrete_bc},
then we have $Pr( x_k \in \mathcal{C},\ \forall 0\leq k \leq T)\geq 1-\gamma- c T$ when the adversary is present. 
\end{theorem}

\begin{proof}
Given condition (i), (ii), and $\bar{\zeta}_s$, by Theorem \ref{th:ftest},  $\|x-\hat{x}_i\|\leq \bar{\zeta}_i$ for each sensor $i\in I\backslash I_a$. In Alg. \ref{alg:FTControl}, $u$ is computed by a nominal controller and $\hat{u}$ is computed by solving \eqref{eq:qcqp}. 
By condition (iii) and Proposition \ref{prop:discrete_bc}, we have $Pr(x[k]\in C, 0\leq k\leq T_d)\geq 1-\gamma- c T_d$.
\end{proof}

\section{Case Study}\label{sec:simulation}

This section evaluates our proposed approach on a UAV delivery system in an urban environment.
The UAV system is based on MATLAB UAV Package Delivery Example \cite{UAVDelivery}. 
The UAV adopts stability, velocity and altitude control modules, rendering its position control dynamics to be:
\begin{multline}
    \label{eq:uav_dyna}
    \begin{bmatrix}
    [x]_1\\
    [x]_2
    \end{bmatrix}_{k+1}
    =
    \begin{bmatrix}
    1 & -4.29\times 10^{-5}\\
    -1.47\times 10^{-5} & 1
    \end{bmatrix}
    \begin{bmatrix}
    [x]_1\\
    [x]_2
    \end{bmatrix}_{k}\\
    +
    \begin{bmatrix}
    0.0019 & -1.93\times 10^{-5}\\
    -2.91\times 10^{-4} & 0.0028
    \end{bmatrix}
    \begin{bmatrix}
    [u]_1\\
    [u]_2
    \end{bmatrix}_{k},
\end{multline}
where $x[k]=[[x]_1,[x]_2]^T$ is the UAV position, $[x]_1$ and $[x]_2$ represent the position of UAV on $X$-axis and $Y$-axis, respectively.
The UAV has one LiDAR sensor and two inertial navigation system (INS) sensors, denoted as INS1 and INS2. The UAV maintains two EKFs associated with each INS sensor to estimate its position at each time $k$, denoted as $\hat{x}_1[k]$ and $\hat{x}_2[k]$, respectively. The system operates in the presence of an adversary who can compromise one of the INS sensors and spoof the LiDAR sensor.


\begin{figure}[b]
    \centering
    \includegraphics[width = 0.33\textwidth]{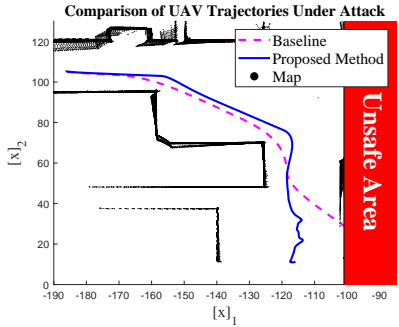}
    \caption{Comparison of trajectories of the UAV when controlled using our proposed approach and the baseline.}
    \label{fig:cs2_ftframe}
\end{figure}

We compare our proposed approach with a baseline utilizing a PID controller with state estimations given by INS1. We first demonstrate how our proposed approach selects sensors via Alg. \ref{alg:ScanReconstruct}  and Alg. \ref{alg:FT-LiDAR_Est} to obtain an accurate state estimation. We consider two attack scenarios. In Scenario I, the adversary compromises INS1 to deviate the measurement by $-20$ meters along the $X$-axis. In Scenario II, the adversary spoofs both the LiDAR sensor and INS1. The adversary biases INS1 sensor by $-20$ meters on $X$-axis and generates a random obstacle in the LiDAR scan within range of $[10,15]$ meters and angle of $[-70,-60]$ degrees.

We present the estimated and actual LiDAR observations under Scenario I in Fig. \ref{fig:s1x1}-\ref{fig:s1x2}. In Fig. \ref{fig:s1x1}, we note that the estimated LiDAR observations $\mathcal{O}(\hat{x}_1,\mathcal{M})$ generated using state estimation $\hat{x}_1$ from INS1 significantly deviates from the actual LiDAR observations (the scan in pink color). The estimated LiDAR observations $\mathcal{O}(\hat{x}_2,\mathcal{M})$ align with the actual one as shown in Fig. \ref{fig:s1x2}, which satisfies the criteria given in Section \ref{sec:ft-state est}. Therefore, we treat INS2 as a trusted sensor while ignoring the measurements from INS1 when computing control input to the UAV.

We next compare the estimated and actual LiDAR observations under Scenario II in Fig. \ref{fig:s2x1}-\ref{fig:s2x2}. The adversary manipulates the  LiDAR observations by injecting a set of false points around position $(5.5, -11.6)$.
In Fig. \ref{fig:s2x1}, we observe a significant drift between the estimated LiDAR observations $\mathcal{O}(\hat{x}_1,\mathcal{M})$ and actual LiDAR observations $\mathcal{O}(x,S)$. In Fig. \ref{fig:s1x2}, the obstacle points contained in sector $c$ generated by the LiDAR spoofing attack are eliminated by Alg. \ref{alg:FT-LiDAR_Est}, and thus the estimated LiDAR observations $\mathcal{O}(\hat{x}_2,\mathcal{M}\backslash c)$ aligns with the LiDAR observations $\mathcal{O}(x,S\backslash c)$. In this case, our proposed fault tolerant estimation indicates that INS1 should be ignored and INS2 can be trusted.

We finally present the trajectories of the UAV with our proposed fault tolerant control (Alg. \ref{alg:FTControl}) and with the baseline. We present the trajectory of our proposed approach in Fig. \ref{fig:cs2_ftframe} as the solid blue line, and the trajectory of the baseline as the dashed pink line.
We observe that our proposed approach ensures the UAV to successfully avoid all obstacles and the unsafe area, whereas the baseline leads to safety violation due to lack of schemes to exclude faulty measurements.
\section{Conclusion}\label{sec:conclusion}

In this paper, we studied the problem of safety-critical control for a LiDAR-based system in the presence of sensor faults and attacks. 
We considered the class of systems equipped with a set of sensors for state and environment observations. 
We proposed a fault tolerant safe control framework for such systems to estimate their states and synthesize a control signal with safety guarantee. To obtain an accurate state estimate, we maintain a set of EKFs computed from different subsets of sensor measurements. 
For each estimate, we construct a simulated LiDAR scan based on the state estimates and an \emph{a priori} known map, and exclude the state estimates that conflict with LiDAR measurements. When the LiDAR scan deviates from all of the state estimates, we remove the sector of the scan with the largest deviation. We proposed a control policy that selects a control input based on the fault tolerant estimate, and proved safety with a bounded probability using a control barrier certificate. 
We validated our proposed method with simulation studies on a UAV delivery system in an urban environment. 
Future work will extend the approach to cases with moving obstacles. 

\bibliographystyle{IEEEtran}
\bibliography{mybib}

\end{document}